\documentclass[journal,onecolumn]{IEEEtran}
\normalsize
\ifCLASSINFOpdf
\else
\fi
\usepackage[T1]{fontenc}
\usepackage[linesnumbered,ruled]{algorithm2e}
\usepackage{stfloats,color}
\usepackage{amsfonts}
\usepackage{amssymb}
\usepackage[cmex10]{amsmath}
\usepackage{graphicx}
\usepackage{setspace}
\usepackage{cite}
\usepackage{array}
\usepackage{subcaption}
\usepackage{times}
\usepackage{epsfig}
\usepackage{latexsym}
\usepackage{epstopdf}
\usepackage{verbatim}
\usepackage{units}
\usepackage{amsthm}
\usepackage{placeins}
\usepackage{afterpage}
\usepackage{dsfont}
\usepackage{soul}
\usepackage{multicol}
\usepackage{multirow}
\usepackage{mathtools}
\usepackage[cmintegrals]{newtxmath}
\newcolumntype{P}[1]{>{\centering\arraybackslash}p{#1}}
\newcolumntype{M}[1]{>{\centering\arraybackslash}m{#1}}

\newcommand{\defeq}{\ensuremath{\triangleq}}

\newtheorem{lemma}{Lemma}
\newtheorem{theorem}{Theorem}
\newtheorem{corollary}{Corollary}

\newtheorem{remark}{Remark}

\begin{document}
%
\title{Mobility-Aware Coded Storage and Delivery}
%
%
%

\author{Emre~Ozfatura and
        Deniz~G{\"u}nd{\"u}z
\thanks{Emre Ozfatura and Deniz G{\"u}nd{\"u}z are withInformation Processing and Communications Lab, Department of Electrical and Electronic Engineering,
Imperial College London Email: \{m.ozfatura, d.gunduz\} @imperial.ac.uk} 
\thanks{This work was supported in part by the Marie Sklodowska-Curie Action SCAVENGE (grant agreement no. 675891), and by the European Research Council (ERC) Starting Grant BEACON (grant agreement no. 725731).}
\thanks{This paper was presented in part at the 2018 ITG Workshop on Smart Antennas in Bochum, Germany.}}

\maketitle

\begin{abstract}
Content caching at small-cell base stations (SBSs) is a promising method to mitigate the excessive backhaul load and delay, particularly for on-demand video streaming applications. A cache-enabled heterogeneous cellular network architecture is considered in this paper, where mobile users connect to multiple SBSs during a video downloading session, and the SBSs request files, or fragments of files, from the macro-cell base station (MBS) according to the user requests they receive. A novel content storage and delivery scheme that exploits coded storage and coded delivery jointly, is introduced to reduce the load on the backhaul link from the MBS to the SBSs. It is shown that the proposed caching scheme, by exploiting user mobility, provides a significant reduction in the number of sub-files required while also reducing the backhaul load when the cache capacity is large. Overall, for practical scenarios in which the number of subfiles that can be created is limited (by the size or the protocol overhead), the proposed coded caching and delivery scheme decidedly outperforms state-of-the-art alternatives.    
\end{abstract} 

\section{Introduction}
\IEEEPARstart{D}ue to the popularity of on demand video streaming services, such as YouTube and Netflix, video dominates the Internet traffic \cite{sandvine,cisco}. A promising solution to mitigate the excessive video traffic and to reduce the video latency in video streaming, particularly in cellular networks, is storing the popular contents at the network edge. There are two prominent approaches used extensively in the literature to reduce the backhaul load in cellular networks; namely, {\em coded storage} and {\em coded delivery}. In a broad sense, coded storage is designed from the  perspective of the users, and allows them to efficiently receive a file from multiple access points without worrying about overlapping bits.  Maximum distance separable (MDS) or fountain codes,  have been studied extensively, for example, in multi-access downlink scenarios, e.g., a static user downloading content from multiple small-cell base stations (SBSs) \cite{femto,SBScache,SBScoop1,SBScoop2}, mobile users (MUs) connecting to different SBSs sequentially to download content \cite{CC.M1,CC.M2,CC.M3,CC.M4,wsa}, or MUs utilizing device-to-device (D2D) communication opportunities \cite{G.M1,G.M2,G.M3,G.M4,G.M5,G.M6}.\\ 
\indent Coded delivery, on the other hand, is designed from the point of the server, which utilizes the caches of the users to seek multicasting opportunities in order to reduce the amount of data it needs to transmit to the users to satisfy their demands \cite{CD.F1,CD.F2,CD.D2D,CD.F3,CD.F4,CD.ND1,CD.ND2,CD.ND3,CD.ND4,CD.ND5,CD.ND6,CD.ND7,CD.ND8,CD.ND9,CD.Dopt4,CD.F5}. Coded delivery schemes consist of two phases. In the {\em placement phase}, files are divided into sub-files, and each user stores a certain subset of the sub-files. In the {\em delivery phase}, the server carefully constructs the multicast messages as XORed combinations of the requested sub-files. Each user recovers its request from the multicasted messages together with its own cache contents. We note that the multicast gain increases with the number of users; that is, the higher the number of users the lower the per-user delivery rate. However, to achieve the promised gain, the number of sub-files has to increase exponentially with the number of users, which is considered as one of the main challenges in front of the implementation of coded delivery in practice \cite{challenge}.\\
\begin{figure*}
    \centering
         \begin{subfigure}[b]{0.47\textwidth}
        \includegraphics[scale=0.35]{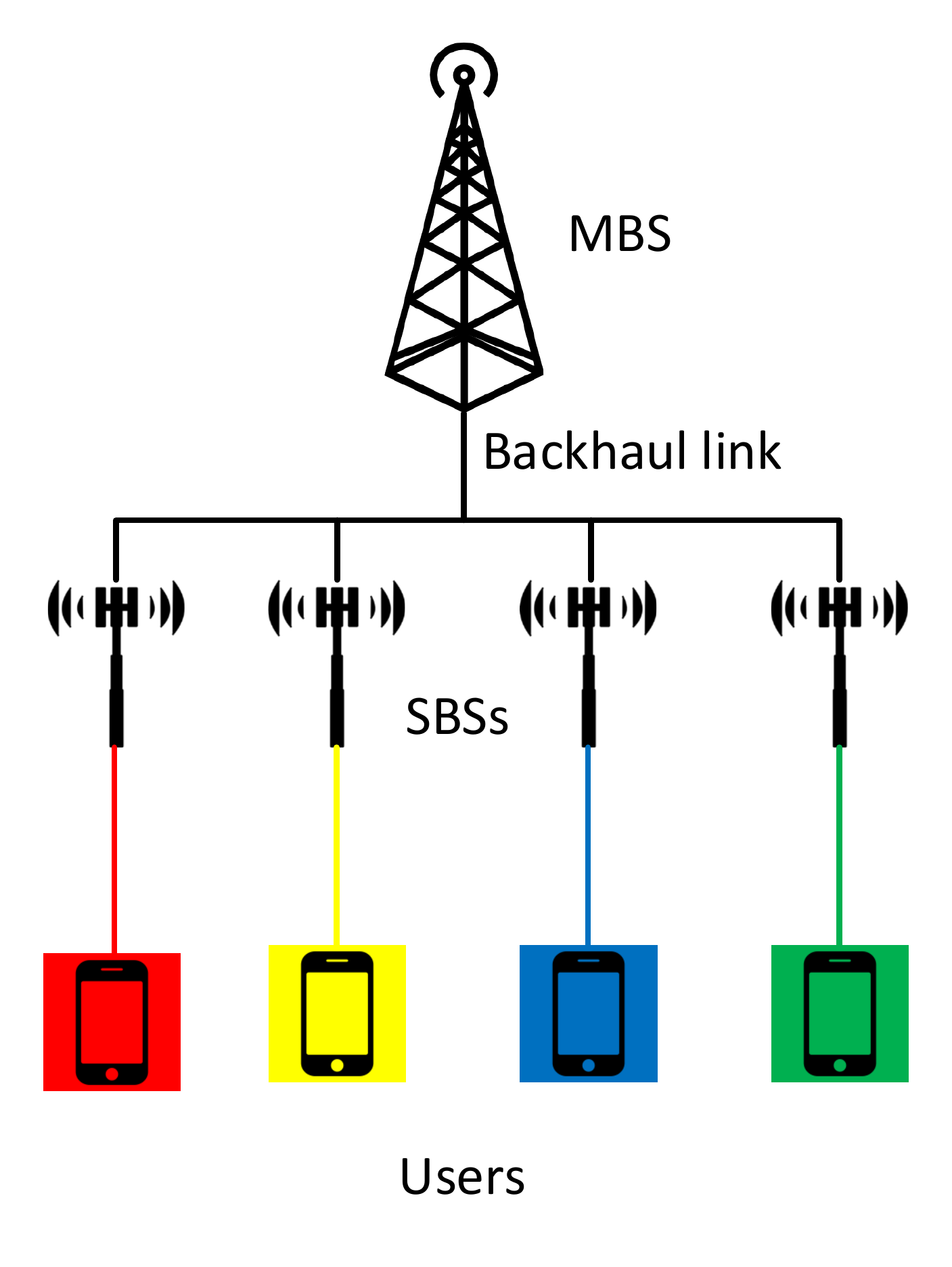}
        \caption{Single access model.}
				\label{model1}
    \end{subfigure}
    \begin{subfigure}[b]{0.47\textwidth}
        \includegraphics[scale=0.35]{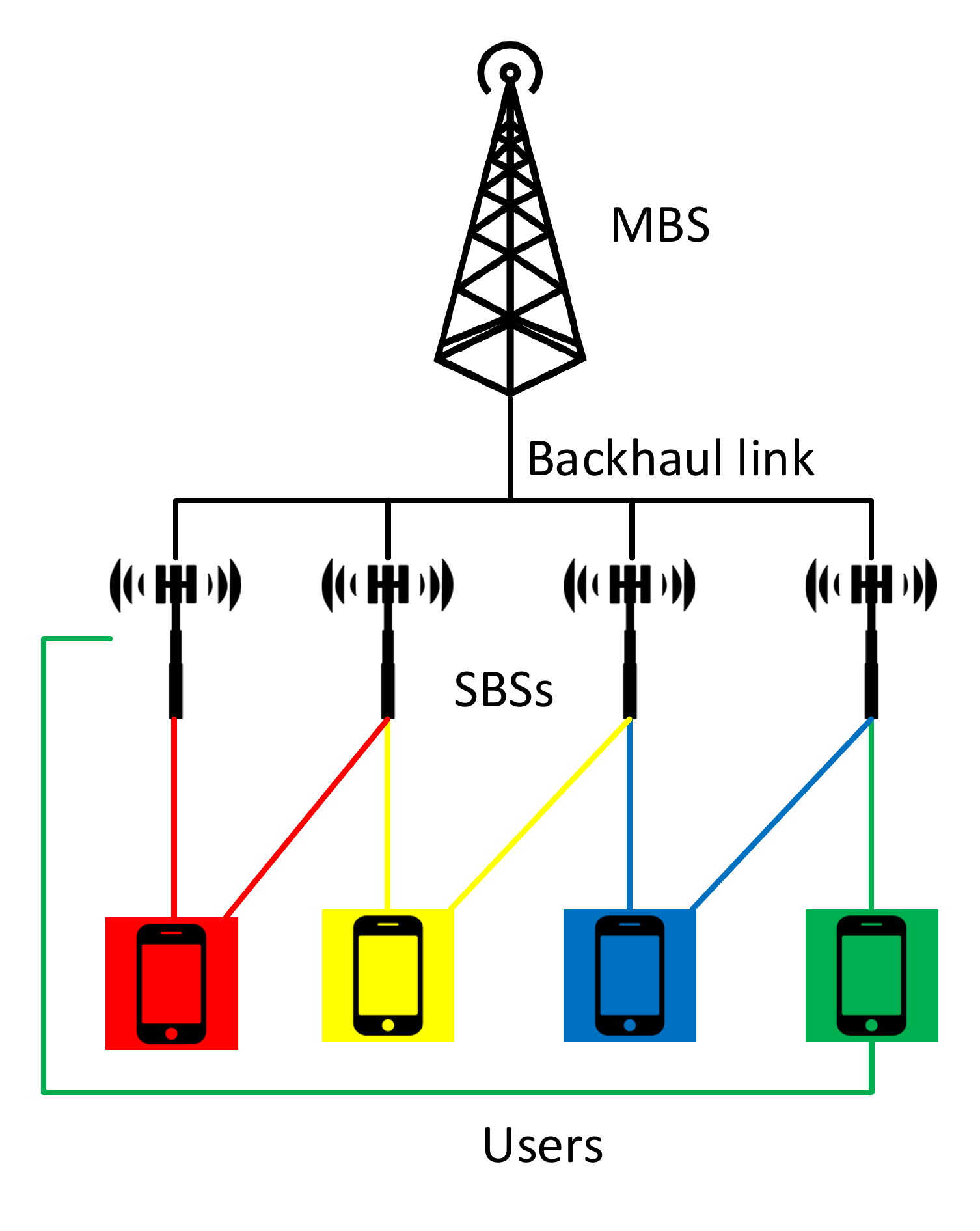}
        \caption{Multi-access model with uniform access pattern.}
				\label{model2}
        \end{subfigure}
				\caption{ Illustration of the static user access models studied in \cite{CD.F1} and \cite{CD.ND4}, respectively.}
		\label{accessmode}
\end{figure*}
\indent To remedy this limitation of coded delivery,  coded caching and delivery designs that require {\em low subpacketizaton levels} have become an active research area. In 
\cite{lowsubpacket3,lowsubpacket4}, the authors show that a particular family of bipartite graphs, namely Ruzsa-Szemer{\'e}di graphs, can be used to  construct a coded caching scheme, in which, for sufficiently large number of caches, $K$, the number of sub-files scales linearly with $K$. Unfortunately, however, the large $K$ assumption limits the practical applicability of the proposed coded caching design. Alternatively, in \cite{lowsubpacket2} a novel coded caching scheme  based on linear block codes is introduced and it is shown that the number of sub-files can be reduced dramatically with a small increase in the delivery rate for any practical value of $K$. It is further shown that, with different linear block codes different delivery rates can be achieved using different  number of sub-files in fixed $K$, allowing some flexibility for implementation. In \cite{lowsubpacket5}, coded caching designs based on linear block codes are constructed by using the so-called {\em placement delivery arrays (PDAs)}. We note that all the aforementioned works seek a coded caching design that reduces the number of sub-files with a minimum sacrifice in the delivery rate.\\
\indent In this paper, we show that the mobility pattern of the users can be utilized to reduce the number of required sub-files by  introducing a novel coded storage and delivery scheme, which is designed to take into account the random mobility patterns of the users. The proposed scheme divides the SBSs into smaller groups according to the mobility patterns of the users, and applies coded delivery to each group of SBSs independently. MDS-coded storage is used to guarantee that the MUs can collect useful information from any of the SBSs they connect to until they recover their requested files. We introduce an efficient grouping strategy via utilizing the analogous well-known frequency reuse pattern problem \cite{reusepattern}. Our contributions in this paper can be summarized as follows:

\begin{itemize}
\item We introduce a two-level coded storage and delivery strategy, and show that for a given MU mobility scenario, optimal coded caching solution for the proposed strategy can be analysed as a cell coloring problem. We further provide the optimal coloring scheme that minimizes the delivery rate.

\item We also show that the proposed mobility-aware scheme can significantly improve the subpacketization-delivery rate trade-off via utilizing the mobility pattern of the MUs, and outperforms state-of-the-art alternatives.

\item Finally, we show that the benefits of the proposed mobility-aware scheme extends also to non-uniform popularity distributions as well as to more general mobility scenarios.
\end{itemize}

\indent In \cite{CD.ND4}, a hierarchical network, in which a macro-cell base station (MBS) serves multiple cache-equipped SBSs through a shared link while each user is connected to $L$ SBSs, is analyzed and it is shown that a lower delivery rate compared to \cite{CD.F1} is achievable (please see Fig. \ref{accessmode} for illustration of the models considered in \cite{CD.F1} and \cite{CD.ND4}). Although it is not highlighted  explicitly in \cite{CD.ND4}, existence of a multi-access pattern, i.e., each user accessing $L$  SBSs, reduces the number of sub-files as well. A similar observation is made in \cite{lowsubpacket6} by leveraging multiple transmission antennas instead of user mobility. Considering a MBS with $L$ transmit antennas, the users are divided into groups of $L$, so that the  channel from the MBS to each group becomes an $L \times L$  channel with $L$ transmit antennas and $L$ single-antenna users. For the delivery phase, a two-level coding scheme is used, where the outer code is designed considering each group of users as a single virtual user, and their $L$ antennas as a single virtual antenna, and the coded caching scheme of \cite{CD.F1} is applied for these $L$ virtual users. Hence, each user in the same group caches the same sub-files. The inner code, for each group, is designed according to the $L \times L$ transmission channel from the server to the $L$ users in this group.\\
\indent The rest of the paper is organized as follows. The system model is introduced in Section \ref{sec:sys}, and the proposed coded storage and delivery  scheme is presented in  Section \ref{sec:sol}. The performance of the proposed scheme is evaluated and compared numerically with the coded delivery scheme in \cite{CD.F1} in Section \ref{sec:results}. Finally, in Section \ref{sec:conc} we conclude the paper with a summary of our main contributions and potential future research directions.\\ 
\indent \textbf{Notations.} Throughout the paper, for positive integer $N$, the set $\left\{1,\ldots,N\right\}$ is denoted by $[N]$. We use $\bigoplus$ to denote the bit-wise XOR operation, while $\binom{j}{i}$ represents the binomial coefficient corresponding to the number of $i$-element subsets of a set with $j$ elements.
 
\section{System Model}\label{sec:sys}
\subsection{Network model}
 We consider a cellular network architecture that consists of one MBS and $K$ SBSs, i.e., $SBS_{1},\ldots,SBS_{K}$. The MBS has access to a content library of $N$ files, $W_{1},\ldots,W_{N}$, each of size $F$ bits. Each SBS is equipped with a cache memory of $MF$ bits. The SBSs are connected to the MBS through a shared wireless backhaul link. Hence, when a user requests a content from a SBS, the SBS first checks its content cache. If the requested content is fully cached, then the SBS directly delivers the corresponding file. If the requested content is not cached at all, or partially cached, then the remaining parts of the content are first transferred from the MBS to the SBS over a backhaul link. We assume that all the demands from the MBS are delivered simultaneously over a shared error-free backhaul connection. \\
 \indent In this paper, we assume that all the files in the library are requested by the users with the same probability, i.e.,  $W_{n}$ is requested by a MU with probability $1/N$, $n\in[N]$. Under this assumption, if $N$ is large compared to $K$, which is the case in realistic scenarios, then it is safe to assume that all the users request a different content. For instance, when $K=30$ and $N=10000$ the probability of each user requesting a different file is $0.975$. This assumption has been widely accepted in the coded caching literature as it also represents the worst case scenario. We also assume that the number of users in the network is limited by the number of SBSs; hence, there are $K$ users in the network, i.e., $U_{1},\ldots,U_{K}$, and the requests of the users are denoted by the demand vector $\mathbf{d}\defeq \left(d_{1},\ldots,d_{K}\right)$.

 \indent The placement phase, during which the caches of the SBSs are filled, takes place before the demand vector $\mathbf{d}$ is revealed. The required delivery rate for the backhaul link, $R(M,N,K,\mathbf{d})$, is defined as the minimum number of bits that must be transmitted over the shared link, normalized by the file size, for a given normalized cache capacity $M$, in order to satisfy all the user demands, $\mathbf{d}$. Since, we assume that each user requests a different file, we simply use $R(M,N,K)$ to denote the worst case delivery rate, instead of $R(M,N,K,\mathbf{d})$.\\
\indent The delivery rate over the backhaul link from the MBS to the SBSs depends on  the access model of the users. In this paper, we are particularly interested in the {\em single access model with mobility}, in which a MU is connected to exactly one SBS at a particular time instant; however, due to mobility, it connects to multiple SBSs over time. To better motivate and explain our model and results, we will first explain the  previously studied access models in the literature, and then provide a detailed explanation of the considered single access model with mobility. We want to remark that the proposed mobility-aware scheme is not limited to {\em uniform-demand} and {\em one MU per SBS} scenario. It can also be applied under {\em non-uniform demand} probabilities; thus, it is compatible with most of the placement schemes introduced for non-uniform demand scenarios \cite{CD.ND1,CD.ND2,CD.ND3,CD.ND4,CD.ND5,CD.ND6,CD.ND7,CD.ND8}, and can be easily employed for  multiple-user per SBS scenario. However, to establish a clear connection between the previously introduced schemes in \cite{CD.ND4} and \cite{CD.F1}, as well as to provide a fair comparison with the schemes proposed in \cite{CD.F1} and \cite{lowsubpacket2} we will introduce our results under this simplifying assumptions. Nevertheless, we later extend our analysis to both non-uniform popularity and multiple-MU per SBS cases later in Section \ref{sec:results}.

\subsection{User access models}
\subsubsection{Static single access model}
In this model, it is assumed that each user is connected to exactly one SBS, as illustrated in Fig. \ref{model1}. This corresponds to the shared link problem introduced in \cite{CD.F1}. The caching and coded delivery method introduced in \cite{CD.F1} works as follows. For $t \triangleq \frac{MK}{N} \in \mathbb{Z}$, in the placement phase, all the files are cached at level $t$; that is, $W_{n}$, $n\in \left[N\right]$, is divided into ${K \choose t}$ non-overlapping sub-files of equal size, and each sub-file is cached by a distinct subset of $t$ SBSs. Then, each sub-file can be identified by a subset $\mathcal{I}$, where $\mathcal{I}\subseteq \left[K\right]$ and $\vert \mathcal{I}\vert =t$, such that sub-file $W_{n,\mathcal{I}}$ is cached by $SBS_{k}$, $k\in\mathcal{I}$. In the delivery phase, for each subset $\mathcal{S}\subseteq \left[K\right]$, $\vert \mathcal{S} \vert=t+1$, all the requests of the SBSs in $\mathcal{S}$ can be served simultaneously by the MBS via multicasting 
\begin{equation}\label{coddel}
\bigoplus_{s\in \mathcal{S}}W_{d_{s},\mathcal{S}\setminus\left\{s\right\}}.
\end{equation}
Thus, with a single multicast message the MBS can deliver $t+1$ sub-files, and achieve a {\em multicasting gain} of $t+1$. Accordingly, the achievable delivery rate is $R(M,N,K)=\frac{K-t}{t+1}$. We emphasize that the promised coded caching gain is obtained by dividing each file into ${K \choose t}$ subfiles, which grows exponentially with $K$. This limits the potential gain in practice for finite-size files.\\
\indent The delivery rate for a cache memory $M$ that results in a non-integer $t$ value can be obtained as a linear combination of the delivery rates of the two nearest $M$ values for which the corresponding $t$ values are integers. This is achieved by {\em memory-sharing} between the caching and delivery schemes for those two $M$ values. In the rest of the paper we will consider integer $t$ values unless otherwise stated.\\
\subsubsection{Static multi-access model}
In this model, each user connects to  multiple SBSs. A particular case of this problem is studied in \cite{CD.ND4}, where each user connects to  $L$ SBSs following a certain cyclic pattern, where user $U_{k}$ connects to $SBS_{k},\ldots,SBS_{k+L-1\mod K}$, $k\in[K]$. The case of $L=2$ is illustrated in Figure \ref{model2}. In \cite{CD.ND4}, the authors divide the SBSs into $L$ groups, where the $l$th group consists of $\mathcal{G}_{l}\defeq\left\{SBS_{k}: k \mod L = l\right\}$. Then, the coded delivery scheme in \cite{CD.F1} is adapted to this setting as follows. In the placement phase each file is divided into $L$ equal-size disjoint fragments, i.e., $W^{l}_{n}$ is the $l$th fragment of file $W_{n}$. Then, for each $l\in\left[L\right]$, all the fragments in $\mathcal{W}^{l}\defeq\left\{W^{l}_{1},\ldots,W^{l}_{N}\right\}$ are cached by the 
SBSs in $\mathcal{G}_{l}$. For the placement of a particular group $\mathcal{G}_{l}$, we use the same caching scheme as in the static single access model with $\hat{K}=K/L$ SBSs, each with a normalized cache size\footnote{The cache capacity is normalized here with respect to the size of a fragment, which is $1/L$ of the original file.} of $\hat{M}=ML$. Therefore, each fragment of each file is cached at level $t\defeq\frac{\hat{K}\hat{M}}{N}=MK/N$, i.e.,  sub-file $W^{l}_{n,\mathcal{I}}$, where $\mathcal{I}\subseteq \left\{k: k \mod L =l\right\}$ and $\vert\mathcal{I} \vert=t$, is cached by $SBS_{k}$, $k\in\mathcal{I}$. Similarly, the coded delivery phase is executed for each $\mathcal{G}_{l}$, $l\in\left[L\right]$, separately. The coded delivery algorithm for this model is given in Algorithm 1, and the corresponding  delivery rate is found as $R(M,N,K)=\frac{K-Lt}{t+1}$. We note that the delivery rate decreases with $L$, the number of SBSs each user connects to. The number of sub-files each file is divided into is $L{K/L \choose t}$ for this scheme, which provides a significant reduction in the subpacketization.
\begin{algorithm}\footnotesize{
		\For{$l=1:L$}{
       \For{$\acute{l}=0:L-1$}{
               \For{$\mathcal{S}\in\left\{k:k \mod L = l\right\}, \vert\mathcal{S}\vert=t$}{
                                $\bigoplus_{s\in \mathcal{S}}W^{l}_{d_{(s-\acute{l})\mod K},\mathcal{S}\setminus\left\{s\right\}}$
				                  }
				            }
				      }
											\caption{Delivery phase for the static multi-access model}
                                }            
\end{algorithm}

\subsubsection{Single access model with mobility}
In this model, MUs connect to different SBSs during the downloading period of a single-file, i.e., a video file or a group of frames, depending on the time scales. But, unlike in the previous model, each MU is connected only to the nearest SBS at any time instant. We consider equal-length time slots, whose duration corresponds to the minimum time duration a MU remains connected to the same SBS. We assume that each SBS is capable of transmitting $B$ bits to a MU within one time slot. Hence, a file of size $F$ bits can be downloaded in $T=\frac{F}{B}$ slots.  We define the mobility path of a user as the sequence of small-cells visited during these $T$ time slots. For instance, for $K=7$ and $T=3$, $SBS_{2},SBS_{3},SBS_{4}$ is one such mobility path. We further assume that during the video downloading session of $T$ time slots each MU is connected to exactly $T$ different SBSs, which we call as the {\em high mobility assumption}.
\subsection{Problem definition}
Our aim is to minimize the normalized delivery rate over the backhaul link under the single access model with high mobility assumption for MUs. In addition to reducing the normalized backhaul delivery rate, we also want to reduce the number of sub-files used in the delivery phase in order to obtain a practically viable caching strategy.\\
\indent We first note that  the single access model with mobility can be  treated similarly to the static single access model in the following way: each file is divided into $T$ disjoint fragments, and each fragment is considered as a separate file so that the size of the file library and the size of the caches are scaled to $NT$ and $MT$, respectively. Then the placement phase is executed as in \cite{CD.F1} according to caching level $t=\frac{KMT}{NT}=KM/N$, and the delivery at each time slot can also be executed as in \cite{CD.F1}, and $R(M,N,K)=\frac{K-t}{t+1}$ is still achievable with $T {K\choose t}$ sub-files. Below we will present an alternative caching and delivery scheme that will reduce the number of required sub-files considerably.\\
 \indent The approach introduced in \cite{CD.ND4} is an efficient method to reduce both the number of sub-files and the normalized delivery rate of the backhaul link, when the users connect to the SBSs in a uniform manner, as described in the previous section. However, this method is not applicable when the users do not follow uniform access patterns. Instead, MDS-coded caching can be employed when the users are mobile, or access the SBSs with non-uniform patterns \cite{nitish,CC.M1,CC.M2}. The key advantage of MDS-coded caching at the SBSs is to reduce the amount of data that need to be cached at each SBS for each file.\\
 \indent Consider the following simple example with $K=4$ SBSs, where a MU can connect to any 3 of them. In this case, each file is divided into 3 fragments, and they are encoded into 4 fragments through a (4,3) MDS code. Each SBS caches a different fragment so that a MU that connects to any three SBSs can recover the file. In this example each SBS needs to cache only one fragment for each file, equivalently, $1/3$ of the original file. Accordingly, under the high mobility assumption with given $T=F/B$, it is sufficient to store only $1/T$ portion of each file at each SBS. Hence, if $M\geq N/T$, via MDS coded storage, the normalized delivery rate over the backhaul link can be reduced to zero, which means that all the user requests can be delivered locally. Otherwise, only $MT/N$ portion of each file can be cached and delivered locally using MDS coded storage at the SBSs. The main drawback of  MDS coded storage is that, coded delivery techniques cannot be applied directly to MDS coded files since the multicasting gain of coded delivery stems from the overlaps among the cached sub-files at  different SBSs. Hybrid designs which leverage coded delivery and coded caching techniques have been previously studied for different network setups \cite{nitish,CD.CP1,CD.CP2}.

\begin{figure}
    \centering
        \includegraphics[scale=0.3]{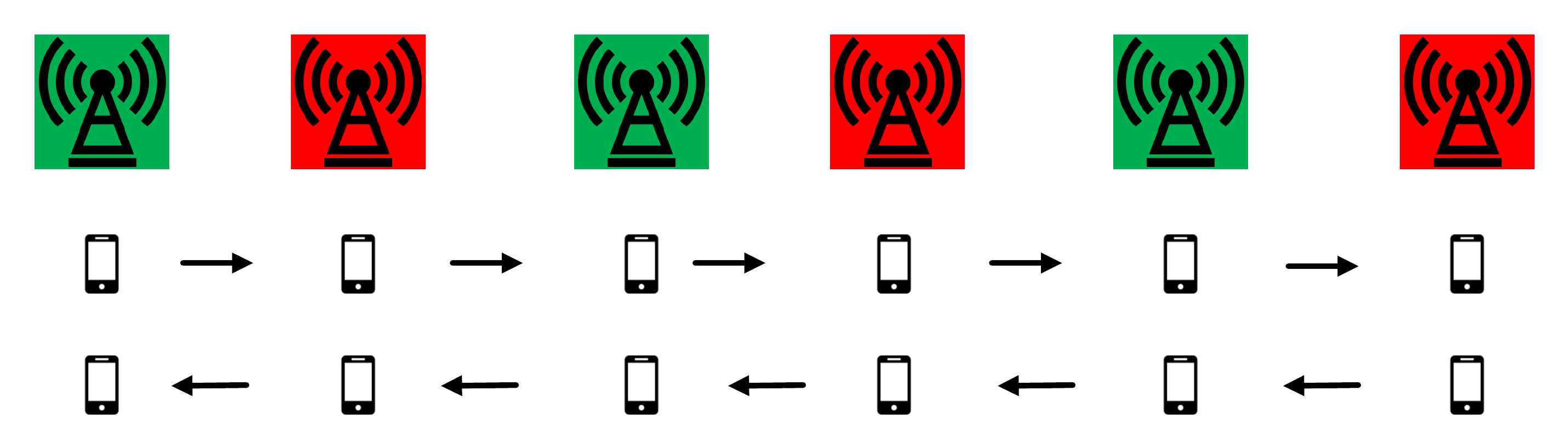}
				\caption{Linear topology: moving along a line. }
				\label{detpath}
\end{figure}   
\section{Solution Approach}\label{sec:sol}
In this section, we introduce a new hybrid design, which utilizes both  MDS-coded caching and coded delivery to satisfy the demands of MUs under the single access model with mobility, and analyze its performance under the high mobility assumption. For the sake of exposition, we first consider a special class of mobility patterns, for which the coded delivery technique in \cite{CD.ND4} can be applied directly.      

\subsection{Special case: Linear topology}
Consider a particular mobility scenario, in which a user's mobility path is determined by its direction and the first SBS it connects to. This can model, for example, MUs on a train connecting to SBSs located by the rail tracks in a known order. In this special case, MUs can be considered as  moving on a line as illustrated in Fig. \ref{detpath}.\\
\indent Although a MU is connected only to the nearest SBS at any time instant, the coded delivery technique introduced in \cite{CD.ND4} for the static multi-access model can be applied in this special case. For given file size $F$ and SBS transmission rate $B$, each file is divided into $T=F/B$ equal-size disjoint fragments, i.e.,  $W_{n}\defeq(W_{n,1},\ldots, W_{n,T})$, $n\in[N]$. Similarly, the set of all SBSs are also divided into $T$ disjoint groups, denoted by $\mathcal{G}_{1},\ldots, \mathcal{G}_{T}$, where SBSs in each group cache only one fragment of each file, using the placement scheme in \cite{CD.F1}. That is, fragment $W_{n,l}$ are cached across SBSs in group $\mathcal{G}_{l}$.\\ 
\indent For a MU to be able to recover all the sub-files of its request, the SBSs should be grouped such that, any mobility path visits exactly one SBS from each group. Grouping of the SBSs can be considered as a coloring problem, where the SBSs are colored using $T$ different colors such that any adjacent $T$ of them have different colors. An example for $T=2$ is illustrated in Fig. \ref{detpath}, where any two neighboring SBSs have different colors. Delivery phase is executed at each time slot separately for each group of SBSs $\left\{\mathcal{G}_{l}\right\}$, $l\in [T]$. Hence, for the special case of linear topology, the achievable delivery rate for the backhaul link is $R(M,N,K)=\frac{K-tT}{1+t}$ with $T\times{K/T\choose t}$ sub-files, where $t=KM/N$, and integer valued as before.
\begin{figure}
    \centering
        \includegraphics[scale=0.8]{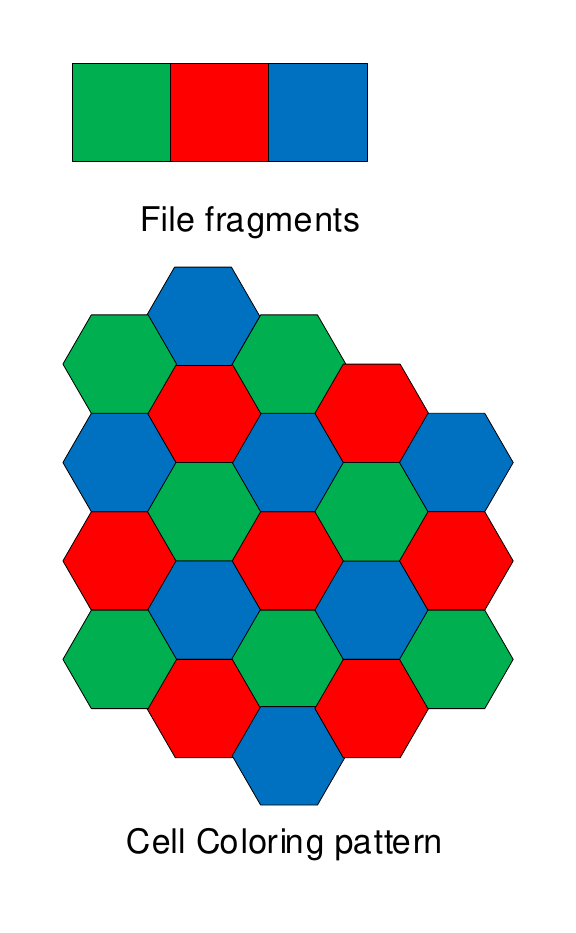}
				\caption{File fragmentation and associated cell coloring.}
				\label{cell}
\end{figure}
\subsection{General Case: Two dimensional topologies}
In this subsection, we consider a general path model in which the users move on a 2D grid, and each SBS covers a disjoint, equal size area with hexagonal shape as illustrated in Fig. \ref{cell}. As opposed to the one dimensional path model, in the 2 dimensional path model it may not be possible to group all the SBSs using only $T$ colors while ensuring that in any path of length $T$ a MU connects to exactly one SBS from each group.\\
\indent For given path length $T$, we will say that the SBSs are {\em $L$-colorable}, if there is a coloring of the SBSs with $L$ colors such that any mobility path of length $T$ consists of $T$ SBSs with different colors. Note that we must have $L\geq T$. The following theorem states the achievable delivery rate over the backhaul link for an $L$-colorable network.
\begin{theorem}
For given values of the variables $N,M,K,T$, and, $t\defeq\frac{KMT}{NL}$ if the network is $L$-colorable, then the following delivery rate over the backhaul link is achievable:
\begin{equation}
R(M,N,K)=\frac{K-tL}{1+t}
\end{equation}
using $T\times{K/L\choose t}$ sub-files, for integer $t$ values.
\end{theorem}

\begin{proof}
In the placement phase, each file is divided into $T$ disjoint fragments with equal size. These are then encoded into $L$ fragments using a $(L,T)$ MDS code. Hence, any $T$ fragments out of the total $L$ is sufficient to decode the original file. Consequently, each group of SBSs (SBSs with the same color) cache a different fragment using the placement scheme in \cite{CD.F1}. The overall delivery phase consists of $T$ identical consecutive delivery steps, each executed in one time slot, such that in each step a coded fragment is delivered to each MU, and having received $T$ fragments at the end of $T$ steps, each MU can recover the original file. To this end, we focus on a single delivery step. The number of SBSs in each group, labeled with the same color, is $\hat{K}=K/L$. If we consider the coded delivery phase for a particular group at a particular time slot, this is identical to the single access model with $\hat{K}$  SBSs each with
a cache memory  of size $\hat{M}=MT$ files; and hence, the corresponding delivery rate is  $\frac{\hat{K}-\hat{K}\hat{M}/N}{1+\hat{K}\hat{M}/N} \frac{1}{T}=\frac{K/L-t}{t+1}\frac{1}{T}$, where $t\defeq\frac{KMT}{NL}$. Accordingly, the overall delivery rate for the backhaul link is found as $R(M,N,K)=\frac{K-tL}{1+t}$.
\end{proof}
 For non-integer $t$ values the following lemma can be used to calculate the corresponding achievable delivery rate. 
\begin{lemma}
If $t=\frac{KMT}{NL}$ is not an integer, then the following rate is achievable by  {\em memory sharing}
\begin{equation}
R(M,N,K)=\left(\gamma\frac{\frac{K}{L}-\left\lfloor t\right\rfloor}{\left\lfloor t\right\rfloor+1}+(1-\gamma)\frac{\frac{K}{L}-\left\lceil t\right\rceil}{\left\lceil t\right\rceil+1}\right)L,
\end{equation}
where $\gamma\defeq\left\lceil t\right\rceil-t$.
\end{lemma}
\indent A simple example for $T=2$ is illustrated in Fig. \ref{cell}. As one can observe from Fig. \ref{cell}, three colors are sufficient to group the SBSs to ensure that in any mobility path of length two a MU always connects to two SBSs with different colors. Hence, in the placement phase of the given example, each file is initially divided into two fragments and these fragments are then encoded  using (3,2) MDS code to obtain 3 coded fragments which are labeled with colors {\color{green}green}, {\color{red}red} and {\color{blue} blue}. All the SBSs in the same group (i.e., those with the same color) cache the fragments that have been assigned the same color. Then, at each time slot, the coded delivery phase is executed for each group of SBSs independently.
\begin{remark}
We remark that the delivery rate achieved for a random path is higher than the one achieved for a deterministic path since the number of colors needed in general is greater than $T$, and it increases further with the number of colors used. Hence, in the single access model with mobility the objective is to identify the minimum $L$ such that the network is $L$-colorable. In the following section  we  study the optimal coloring strategy and the corresponding backhaul delivery rate.
\end{remark}
\subsection{Cell coloring}

\begin{figure*}
\centering
     \includegraphics[scale=0.4]{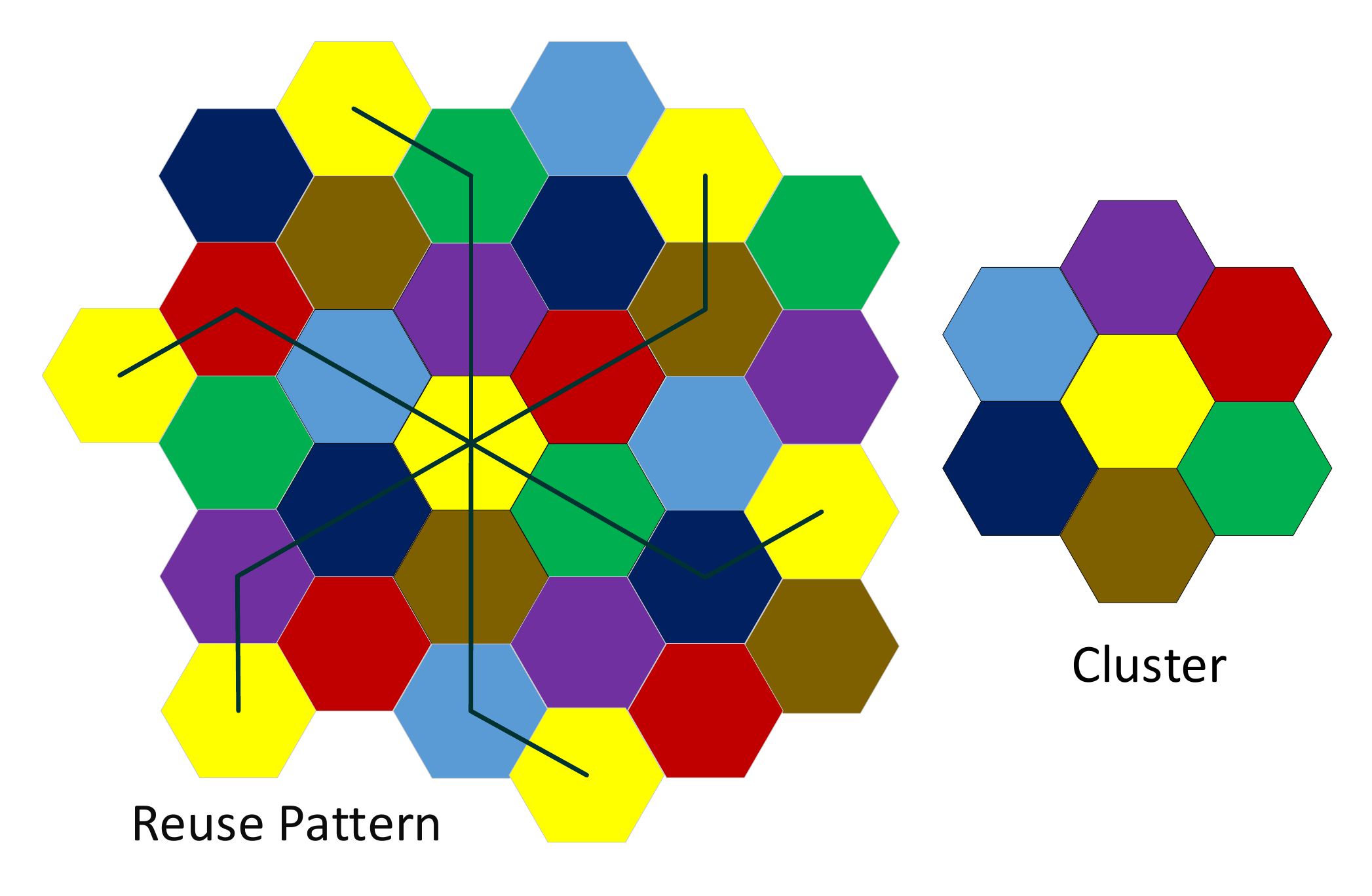}
				\caption{Frequency reuse pattern with i=2, j=1 and the corresponding cluster.}
		\label{reusecluster}
\end{figure*}
\indent For a given mobility path of length $T$, our objective is to color the cells using the minimum number of distinct colors while ensuring that, in any mobility path each color is encountered at most once. In general, for any given cell structure, and mobility length $T$, the cell coloring problem can be modeled as a vertex coloring problem. Consider $K$ SBSs with disjoint cells. We can consider each cell as a vertex of a graph $G$, and add an edge between vertices $k$ and $j$ if  there is a mobility path of length $T$ that contains both  cell $k$ and $j$. Once the graph $G$ is constructed, the chromatic number $\gamma(G)$ of this graph gives the minimum $L$ such that the network is $L$-colorable. This vertex coloring problem, i.e., finding the chromatic number $\gamma(G)$, is NP-hard \cite{complexity}. Hence, finding the optimal coloring, the minimum $L$, in a large network may not be feasible. Therefore in the scope of this paper we focus on the scaling behavior of the cell coloring problem i.e., for the given mobility path length $T$, what is the minimum $L$ as $K$ goes to infinity.\\  
\indent Let us limit our focus on hexagonal cells first. This problem is analogous to the well known {\em frequency reuse pattern} problem in cellular networks \cite{reusepattern}. In this problem, the same frequency is allocated to multiple cells to efficiently use the limited available spectrum while minimizing interference. We note that utilization of frequency reuse patterns in the coded delivery framework has been previously studied in \cite{CD.D2D} for limiting the interference in device-to-device communication.
The cells serving in the same frequency are called the {\em co-channel cells}. The co-channel cell locations are determined according to a given distance constraint (the distance between the center of two co-channel cells). In \cite{reusepattern}, a frequency reuse pattern (or, equivalently, a co-channel cell pattern) is defined via the integer-valued shift parameters $i$  and $j$ in the following way: starting from a cell, ``move $i$ cells along any chain of  hexagons; turn counter-clockwise 60 degrees; move $j$ cells along the chain that lies on this new heading". A frequency reuse pattern example with $i=2$ and $j=1$ is illustrated in  Fig. \ref{reusecluster}. When co-channel cells are identified with a same color, then the pattern of cells with different colors, so that the whole network is a repetition of this pattern, is called the {\em cluster}, which is illustrated in  Fig. \ref{reusecluster}. It is shown that using a reuse pattern with shift parameters $i$ and $j$, the two nearest co-channels are separated with a distance $D=\sqrt{3C}$ (scaled with the cell diameter), where $C=i^{2}+j^{2}+ij$ is the {\em cluster size} which refers to the total number of different frequencies used in the network.\\
\indent We remark that when the nearest co-channel cells are separated with a distance $D=\sqrt{3C}$ according to the reuse pattern with shift parameters $i$ and $j$,  a user in   a particular cell should visit at least $i+j$ (including the current cell) cells to reach the nearest co-channel cell, and by definition no two of these cells can be co-channel cells. Therefore,  the frequency reuse pattern problem is analogous to our problem, where the length of the mobility path $T$ is equivalent to $i+j$, and the cluster size $C$ is equivalent to the number of colors $L$. In our problem, we want to minimize the number of colors $L=T^{2}-ij$ for a given mobility length $T=i+j$. Hence,  we use the reuse pattern $(i,j)$, with $i=\left\lceil \frac{T}{2}\right\rceil$ and $j=\left\lfloor \frac{T}{2} \right\rfloor$ to minimize the number of colors $L$. The scaling behavior of the clusters with respect to $T$ is illustrated in Fig. \ref{scaling}.
\begin{remark}
We observe that when the reuse pattern $(i,j)$, with $i=\left\lceil \frac{T}{2}\right\rceil$ and $j=\left\lfloor \frac{T}{2} \right\rfloor$ is used for a given $T$, then in the corresponding cluster, it is possible to reach a cell from any other cell in $T$ steps so that each cluster corresponds to complete graph. 
\end{remark}

\begin{figure*}
         \begin{subfigure}[b]{0.47\textwidth}
          \centering
     \includegraphics[scale=0.3]{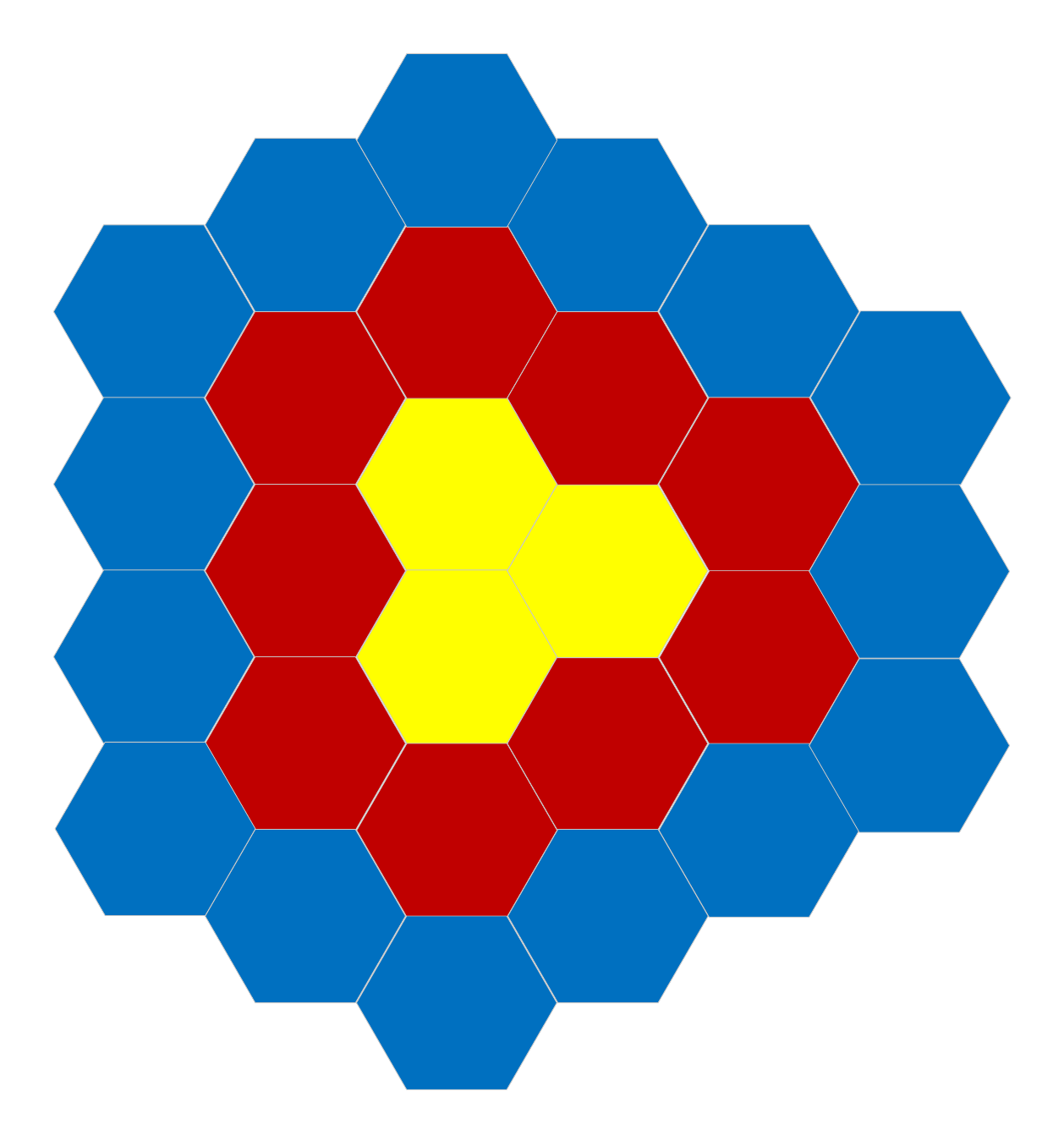}
				\caption{Clusters for $T=2$, $T=4$ and $T=6$ are illustrated with yellow, yellow and red, and  all three colors, respectively.} 
				\label{scalingeven}
    \end{subfigure}
    \begin{subfigure}[b]{0.47\textwidth}
     \centering
        \includegraphics[scale=0.3]{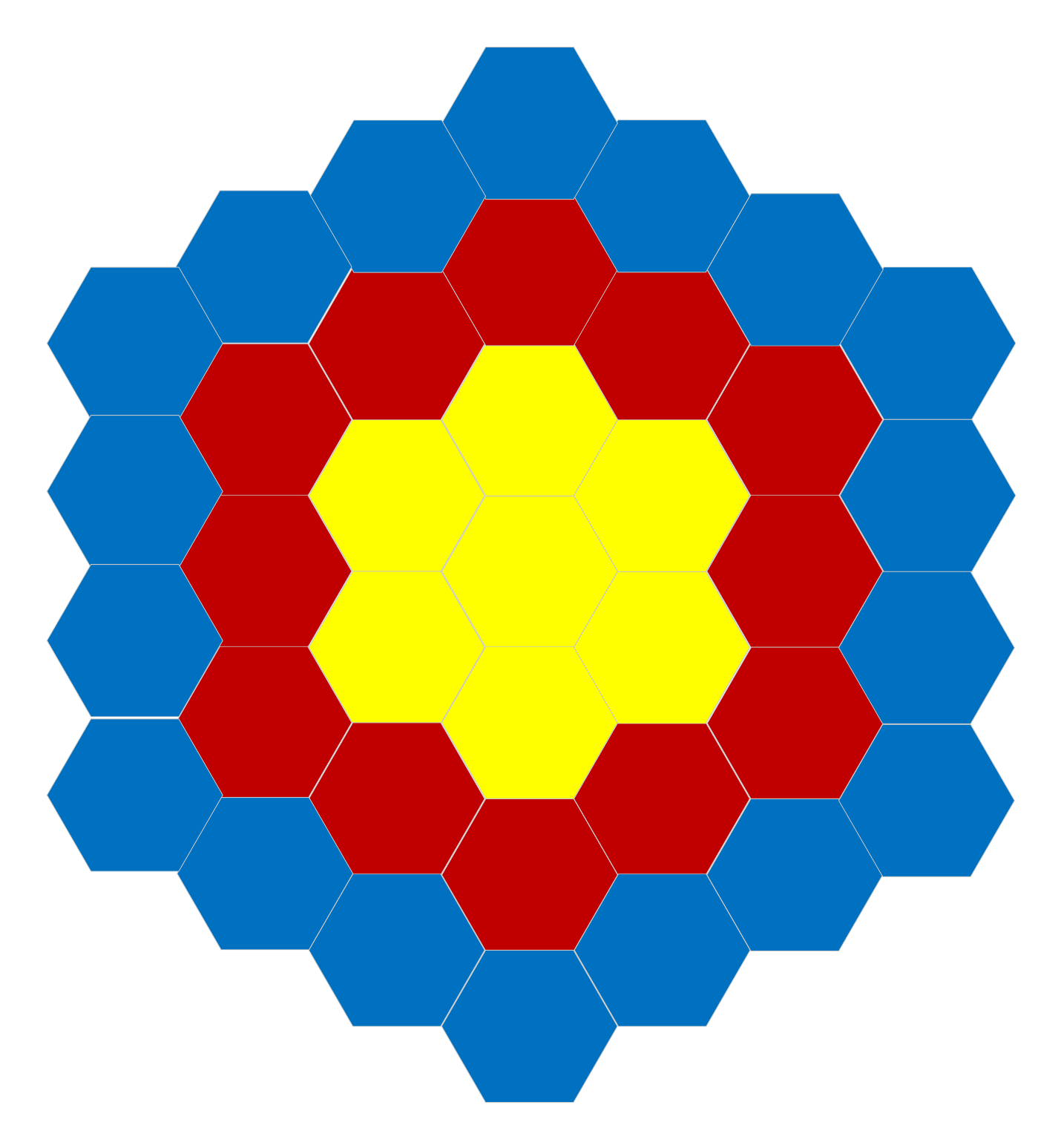}
        \caption{Clusters for $T=3$, $T=5$ and $T=7$ are illustrated with yellow, yellow and red, and all three colors, respectively.} 
				\label{scalingodd}
        \end{subfigure}
				\caption{Scaling behavior of clusters for hexagonel cells.}
		\label{scaling}
\end{figure*}
\begin{theorem}
For a network of SBSs with  hexagonal cells, and a given mobility path length $T$, the minimum $L$ such that the network is $L$-colorable is given by
\begin{equation}
L_{min}= 
    \begin{cases}
     3n^{2}, &  \text{if } T=2n, \\
      3n^{2}+3n+1,   &  \text{if } T=2n+1, 
    \end{cases}
\end{equation}
for some positive integer $n$.
\end{theorem}
\begin{proof}
It is clear that  $L_{min}$ is achievable by using the proper reuse pattern explained. On the other hand, since the cluster corresponds to a complete graph, its chromatic index is equal to cluster size $L_{min}$ which implies that it is not possible to color network with less than $L_{min}$ different colors.
\end{proof}
We remark that this analysis can be extended to different network topologies. For instance, if we consider a square grid, the given reuse pattern can be modified by simply using 90-degree turns instead of 60. An example of a reuse pattern with i=2, j=2 is illustrated in Fig. \ref{reusesq} for a square grid.  
\begin{corollary}
For a network of SBS with  square cells, and a given mobility path length $T$, the minimum $L$ such that the network is $L$-colorable is given by
\begin{equation}
L_{min}= 
    \begin{cases}
     2n^{2}, &  \text{if } T=2n, \\
      2n^{2}+2n+1,   &  \text{if } T=2n+1, 
    \end{cases}
\end{equation}
for some positive integer $n$.
\end{corollary}

\begin{figure}[t]
    \centering
        \includegraphics[scale=0.4]{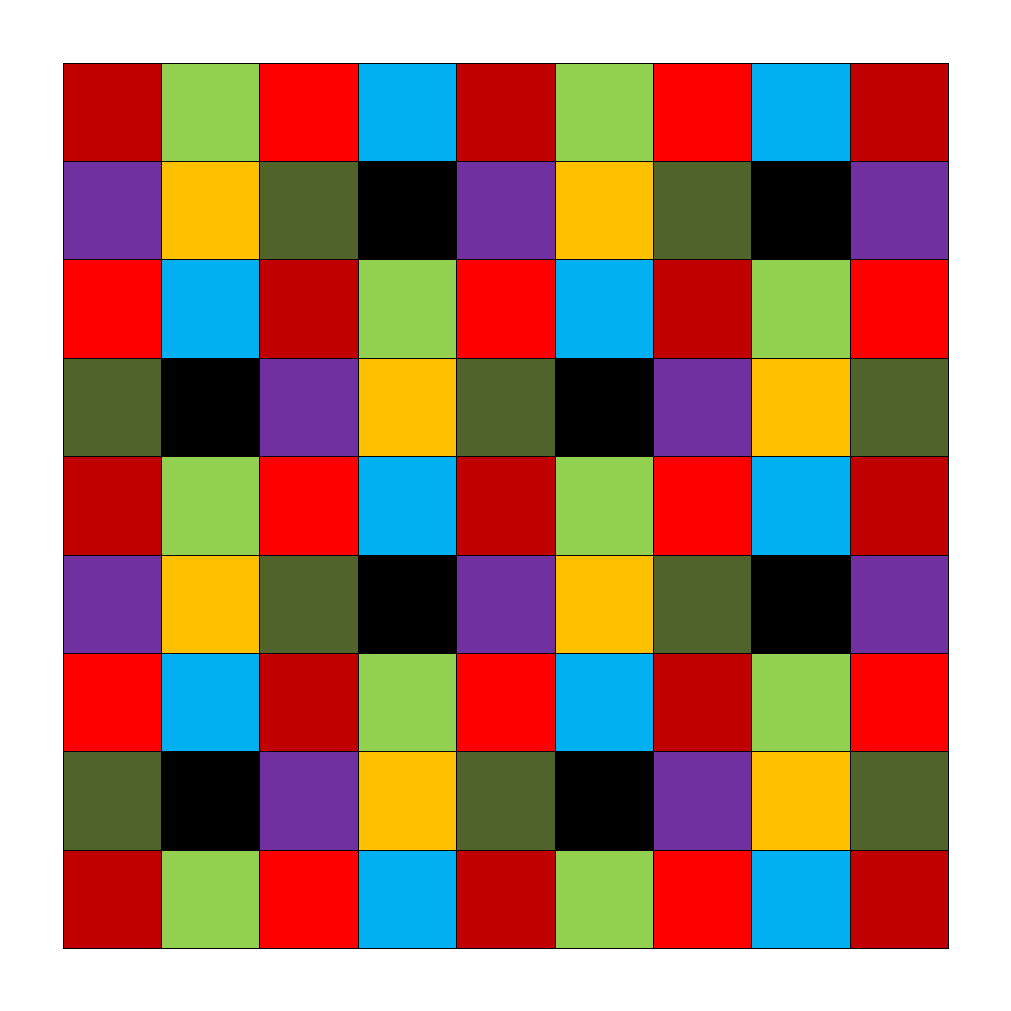}
				\caption{Cell coloring in a square grid according to the reuse pattern with i=2, j=2.}
				\label{reusesq}
\end{figure}
The scaling behavior of the clusters in a square cell topology is illustrated in  Fig. \ref{scaling2}. We remark that as the number of neighboring cells increases, the cluster size also increases. For instance, for a  mobility path length $T=4$,  the corresponding cluster size is 12 for hexagonal cells whereas  it is 8 for the square cells. Recall that the delivery rate of the mobility-aware coded delivery scheme increases with the cluster size $L$. Hence, we can conclude that the mobility-aware scheme performs better when the number of cells a MU can move into at each step is smaller, or equivalently, the mobility pattern has less uncertainty.

\begin{figure*}
         \begin{subfigure}[b]{0.47\textwidth}
          \centering
     \includegraphics[scale=0.4]{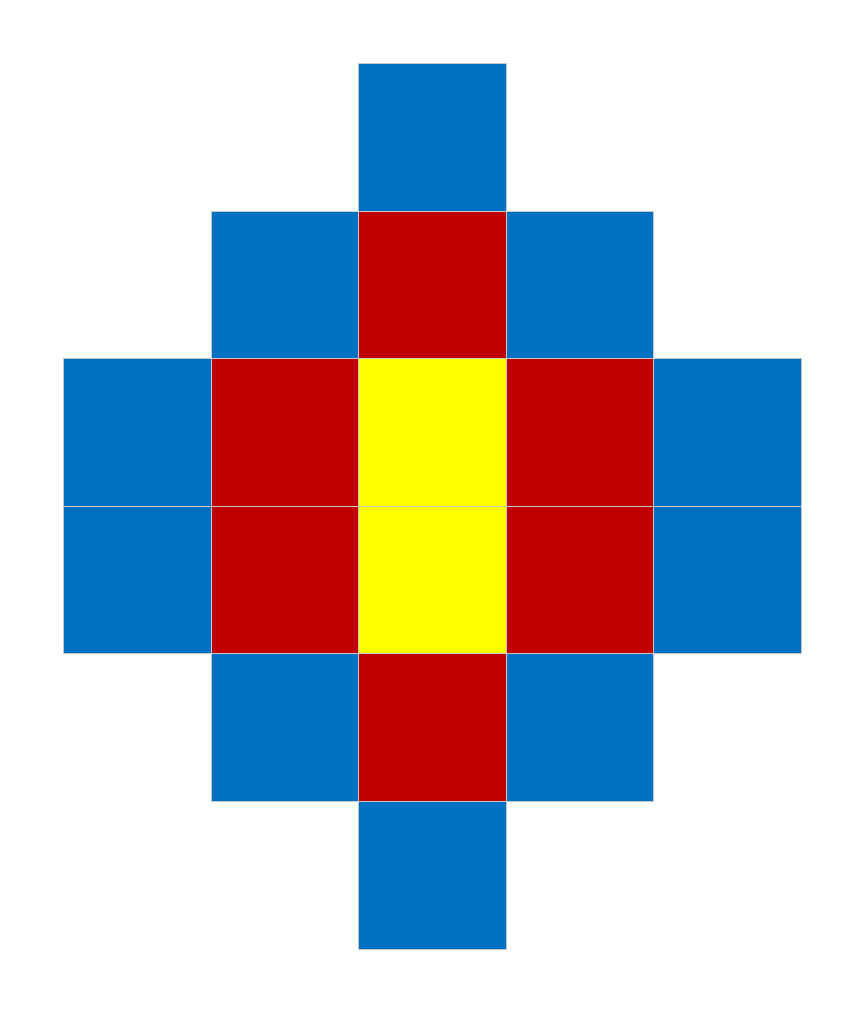}
					\caption{Clusters for $T=2$, $T=4$ and $T=6$ are illustrated with yellow, yellow and red, and  all three colors, respectively.} 
				\label{scaling2even}
    \end{subfigure}
    \begin{subfigure}[b]{0.47\textwidth}
     \centering
        \includegraphics[scale=0.2]{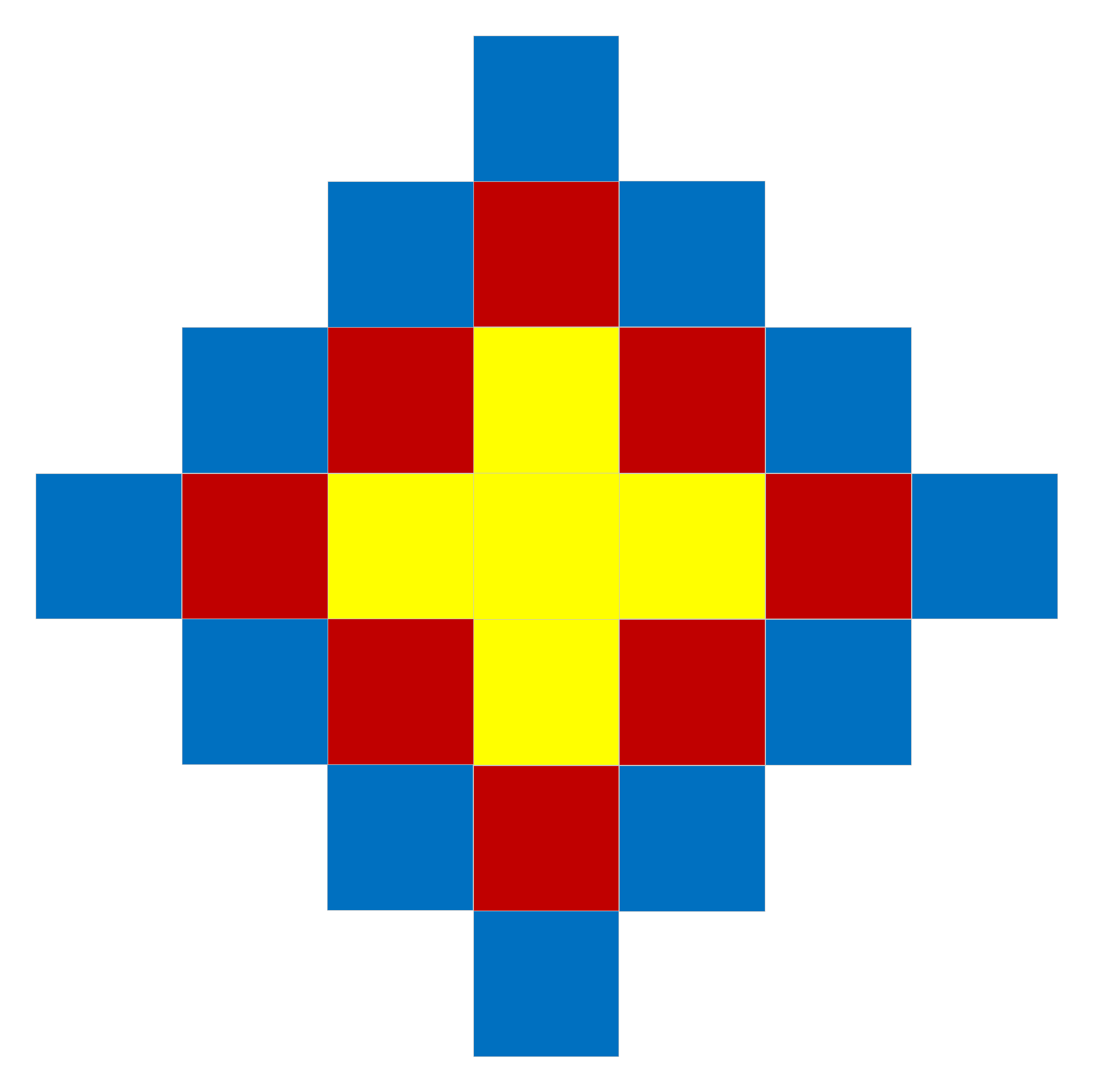}
        \caption{Clusters for $T=3$, $T=5$ and $T=7$ are illustrated with yellow, yellow and red, and  all three colors, respectively} 
				\label{scaling2odd}
        \end{subfigure}
				\caption{Scaling behavior of clusters for square cells.}
		\label{scaling2}
\end{figure*}


 \section{Numerical Results} \label{sec:results}
For the simulations, we consider two  network topologies with $K=24$ and $K=48$ SBSs of hexagonal shapes, respectively. We consider a mobility path of length $T=2$. Hence, the cells are colored according to the reuse pattern $(i=1,j=1)$ with a total of $L=3$ colors as in Fig. \ref{cell}.
 We compare the performance of our mobility-aware coded delivery scheme with the coded delivery scheme of \cite{CD.F1} in terms of two metrics: the number of required sub-files and the normalized backhaul delivery rate. For each topology, we analyze the performance of these schemes for two different storage capacities  of $M/N=1/4$ and $M/N =1/8$, respectively. The numerical results are presented in Table \ref{simres}.\\
\indent In the first group of simulations we analyze the network topology  with $K=24$ SBSs and observe that our mobility-aware coded delivery scheme reduces the number of sub-files dramatically. When, $M/N=1/8$ our mobility-aware coded delivery scheme has $12.5\%$ increase in the delivery rate, while reducing the number of sub-files by approximately $1/72$. The more interesting results are observed when the storage capability is higher, i.e., $M/N=1/4$. In this case, the proposed mobility-aware coded delivery scheme outperforms the original  coded delivery scheme in both performance metrics. At the first glance this might be counterintuitive since there is a trade-off between the delivery rate and the number of sub-files \cite{lowsubpacket2}. However, the mobility-aware approach not only utilizes the multicasting gain, but  also the multi-access gain, which is clearly visible in the deterministic path scenario. In this network setting the number of required sub-files goes from $269000$ down to $140$. \\    
\indent In the second group of simulations, we analyze the network topology  with $K=48$ SBSs. When  $M/N=1/8$ our mobility-aware coded delivery scheme results in a $20\%$ increase in the delivery rate, while reducing the number of sub-files by approximately four orders of magnitude. Hence, thanks to the proposed approach, coded delivery with caching can be practically realizable with only $20\%$ increase in the delay.

At this point, one can argue that the number of sub-files could also be reduced by simply clustering the SBSs to obtain two sub-networks with $K/2$ SBSs, and then applying the coded delivery scheme to each sub-network independently. Indeed, the clustering approach could reduce the number of sub-files significantly; however, it leads to a further increase in the backhaul delivery rate. The results with the clustering approach, assuming two clusters, each consisting of $K/2=24$ SBSs, are included in Table  \ref{simres}. We note that when there are two clusters, the corresponding delivery rate is simply the sum of the delivery rates corresponding to each cluster. Hence, the coded delivery scheme with two clusters uses the same number of sub-files as the coded delivery scheme for the network topology with $K=24$ SBSs; however, the delivery rate is doubled. One can easily observe that for both $M/N=1/8$ and $M/N=1/4$ our mobility-aware coded delivery scheme outperforms the coded delivery scheme with two clusters in terms of both performance metrics. We also observe that the mobility-aware coded delivery approach becomes more efficient compared to the other two schemes, particularly when the storage capacity is high. To highlight this fact, for $T=2$, consider the extreme point $M/N=1/2$. In this case the backhaul delivery rate reduces  to zero, while the number of subfiles is only two.\\
\indent We remark that a more sophisticated scheme, such as the one utilizing the erasure code design in \cite{lowsubpacket2}, can be  also applied to seek a balance between the number-of sub-files and the delivery rate. To this end, we consider the  scenario in Example 9 in \cite{lowsubpacket2}, where there are $K=60$ SBSs with hexagonal shapes and $M/N=1/5$. Similarly to the previous setup we consider $T=2$. In this setup, the original coded delivery scheme achieves a slightly lower delivery rate compared to the mobility-aware scheme with approximately $10^{7}$ times more sub-files. To illustrate the efficiency of the mobility-aware scheme we can limit the number of files to be less than $10^{7}$ and then compare the achievable delivery rates. Performances of the clustering method and the block code design in \cite{lowsubpacket2}, under the subpacketization constraint  are shown in Table \ref{simres2}. One can observe that the proposed mobility-aware caching scheme  outperforms the clustering scheme and the block code design both in terms of the delivery rate under the subpacketization constraint. At this point, it is worth emphasizing that, performance of the mobility-aware scheme depends on the mobility length $T$, thus to reduce the number of sub-files further, for each coded fragment placement scheme in \cite{lowsubpacket2} is used instead of the one in \cite{CD.F1}.

\begin{table*}
    \centering
    \begin{tabular}{| M{3cm} | M{7.5cm} | M{2cm} |M{2.5cm}|}
    \hline
     Storage capacity ($M/N$) & Coded delivery method and network scenario & Number of sub-files &Normalized Delivery rate\\ \hline
     \multirow{2}{*}{$\frac{1}{8}$} & Coded delivery \cite{CD.F1}, for $K=24$ & 4048 & 5.25\\
     & Mobility-aware coded delivery for $K=24$ & 56 & 6 \\ \hline 
 		 \multirow{2}{*}{$\frac{1}{4}$} & Coded delivery \cite{CD.F1}, for $K=24$  & $2.69\times 10^{5}$ & 2.57\\     
      & Mobility-aware coded delivery for $K=24$  & $140$ & 2.4 \\ \hline 
           \multirow{3}{*}{$\frac{1}{8}$} & Coded delivery \cite{CD.F1}, for $K=48$   & $2.45\times 10^{7}$ & 6\\
     & Coded delivery for $K=48$  with clustering & $4048$ & 10.5\\     
     & Mobility-aware coded delivery for $K=48$  & $3640$ & 7.2\\ \hline 
 		 \multirow{3}{*}{$\frac{1}{4}$} & Coded delivery \cite{CD.F1}, for $K=48$  & $1.39\times 10^{11}$ & 2.77\\ 
         & Coded delivery for $K=48$ with clustering & $2.69\times 10^{5}$ & 5.14\\
      & Mobility-aware coded delivery for $K=48$ & $2.57 \times 10^{4}$ & 2.66 \\ \hline 
    \end{tabular}
	\caption{Comparison of the proposed mobility-aware coded storage and  delivery scheme with the conventional coded delivery scheme of \cite{CD.F1} and a coded delivery scheme with clustering in terms of the number of required sub-files and the normalized delivery rate.}
		\label{simres}		
\end{table*}
 \begin{table*}
\begin{center}
    \begin{tabular}{ |  c | c |c |}
    \hline
      Coded delivery method  & Number of sub-files &Normalized Delivery rate\\ \hline
       Coded delivery \cite{CD.F1} & $2.8\times 10^{12}$ & 3.69\\ \hline 
      Mobility-aware coded delivery & 251940 & 4 \\ \hline 
 		  Coded delivery using (12,8) block code \cite{lowsubpacket2} & $2.34\times 10^{6}$ & 5.33\\ \hline 
           Coded delivery with two clusters & $1.18\times 10^{6}$ & 6.85\\ \hline
     \end{tabular}		
 \end{center}
 \caption{Comparison of the proposed mobility-aware coded storage and  delivery scheme with the conventional coded delivery scheme of \cite{CD.F1} and a coded delivery scheme using block code designs terms of the number of required sub-files and the normalized delivery rate.}
 \label{simres2}
\end{table*}
\subsection{Non-uniform file popularity}
The introduced delivery rate analysis for the proposed mobility-aware scheme is based on uniform popularity assumption, which reflects the worst case performance. However, in real on-demand content streaming applications, uniform popularity assumption is not realistic; indeed, many works analyzing the content popularity statistics \cite{vs1,VS2,VS3,VS4} show that only a small fraction of the contents are requested frequently, and the content popularity statistics can be modeled by  Zipf and Weibull distributions with appropriate parameters. Performance of coded delivery schemes can be further improved via utilizing the popularity distribution by allocating more cache memory to popular files, instead of sharing the local cache memory between the files uniformly. In the case of non-uniform popularity, the objective is to minimize the expected delivery rate for a given popularity distribution.\\
\indent In the literature, several different schemes have been introduced for the  non-uniform demand problem \cite{CD.ND1,CD.ND2,CD.ND3,CD.ND4,CD.ND5,CD.ND6,CD.ND7,CD.ND8}. A simple yet efficient scheme, proposed in \cite{CD.ND1}, groups files according to their popularities. One particular implementation of this scheme is the {\em file removal strategy}, in which the whole library  is divided into two groups, namely {\em popular} and {\em unpopular} files, and  only the files in the first group are cached with equal cache memory allocation. We remark that, under the file removal strategy, caching less files (those more popular ones) with a larger value of $t$, decreases the delivery rate for the cached files, but also increases the likelihood of requesting an uncached file. Hence, the optimal cache placement strategy decides the number of files to be cached, and once the  expected delivery rate is written in terms of the number of cached files, it can be easily optimized. For our analysis, we utilize this strategy for the non-uniform demand case.\\
\indent Let $\Pi(K,M,N,N_{c})$ denote the cache placement policy, where $N_{c} \leq N$ denotes the number of cached files. Under a cache placement policy $\Pi$, the required delivery rate for a given demand vector $\mathbf{d}$ can be written as the sum of two rate functions corresponding to cached and uncached files, i.e.,
\begin{equation}
R_{total}(\Pi,\mathbf{d})=R_{c}(M,N_{c},K,\mathbf{d})+R_{u}(N-N_{c},\mathbf{d}).\label{rtotal}
\end{equation}
The delivery rate corresponding to uncached files, $R_{u}$, is simply equal to the number of requests for unchached files, while the delivery rate corresponding to cached files can be written as,
\begin{equation}
R_{c}(M,N_{c},K,\mathbf{d})=R(M,N_{c},K)- \sum^{T}_{\tau=1}\sum^{L}_{l=1}\left(\gamma\frac{ {N^{(l,\tau)}_{u}(\mathbf{d}) \choose \left\lfloor t\right\rfloor+1}}{{\hat{K} \choose \left\lfloor t\right\rfloor}}+ (1-\gamma)\frac{{ N^{(l,\tau)}_{u}(\mathbf{d}) \choose \left\lceil t\right\rceil+1}}{{\hat{K} \choose \left\lceil t\right\rceil}}
\right)\frac{1}{T},\label{rtotal2}
\end{equation}
where $N^{(l,\tau)}_{u}$ denotes the number of unchached requests at time slot $\tau$ in cluster $l$, and $t\defeq\frac{KMT}{N_{c}L}$. Then, the expected total delivery rate, over the demand distribution, can be formulated as 
\begin{align}
\mathbb{E}_{\mathbf{d}}\left[R_{total}(\Pi,\mathbf{d})\right]= & R(M,N_{c},K)- \mathbb{E}_{\mathbf{d}}\left[\sum^{T}_{\tau=1}\sum^{L}_{l=1}\left(\gamma\frac{ {N^{(l,\tau)}_{u}(\mathbf{d}) \choose \left\lfloor t\right\rfloor+1}}{{\hat{K} \choose \left\lfloor t\right\rfloor}}+ (1-\gamma)\frac{{ N^{(l,\tau)}_{u}(\mathbf{d}) \choose \left\lceil t\right\rceil+1}}{{\hat{K} \choose \left\lceil t\right\rceil}}
\right)\right]\frac{1}{T}+\underbrace{\mathbb{E}_{\mathbf{d}}\left[ \sum^{T}_{\tau=1}\sum^{L}_{l=1}N^{(l,\tau)}_{u}(\mathbf{d})\right]\frac{1}{T}}_{K\times (1-p_{c}(N,N_{c}))}\label{rtotal31}\\
& \leq R(M,N_{c},K) + K \times (1-p_{c}(N,N_{c})).\label{rtotal32}
\end{align}
Let $\mathbf{P}=[p_{1},\ldots,p_{N}]$ denote the popularity vector,  where $p_{n}$ is the probability of file $W_n$ being requested. We assume, without loss of generality, that $p_{1}\geq p_{2}\geq\ldots\geq p_{N}$ and we introduce $p_{c}(N,N_{c})\defeq\sum^{N_{c}}_{n=1} p_{n}$ to denote the probability of requesting a cached file. For given parameters $K$, $M$, $N$, $L$, $T$ and $\mathbf{P}$, our objective is to find an optimal cache placement policy $\Pi^{\star}$ that minimizes (\ref{rtotal31}). We note that, in general, the second term in (\ref{rtotal31}) is negligible\footnote{The ratio between the first two term scales with $(K/N_{u})^{t+1}$, thus when $t\geq 2$ the first term dominates the second term unless $K/N_{u}<2$, but when $K/N_{u}<2$, then $R_{u}$ becomes the dominant factor. Hence, this assumption is reasonable in general.} compared to the other terms; besides, the second term in (\ref{rtotal31}) has no closed form expression which makes it difficult to find the optimal cache placement policy  $\Pi^{\star}$. Hence, for the popularity-aware cache placement strategy, we minimize the upper bound given in (\ref{rtotal32}) instead of (\ref{rtotal31}). One can easily observe that finding the optimal cache placement policy  $\Pi^{\star}$ is simply finding the value of $N_{c}$ that minimizes (\ref{rtotal32}). Therefore, the optimal cache placement strategy can be found easily by searching over all possible values of $N_{c}$, with a computational complexity of $\mathcal{O}(N)$.\\
\begin{figure*}
         \begin{subfigure}[b]{0.47\textwidth}
          \centering
     \includegraphics[scale=0.5]{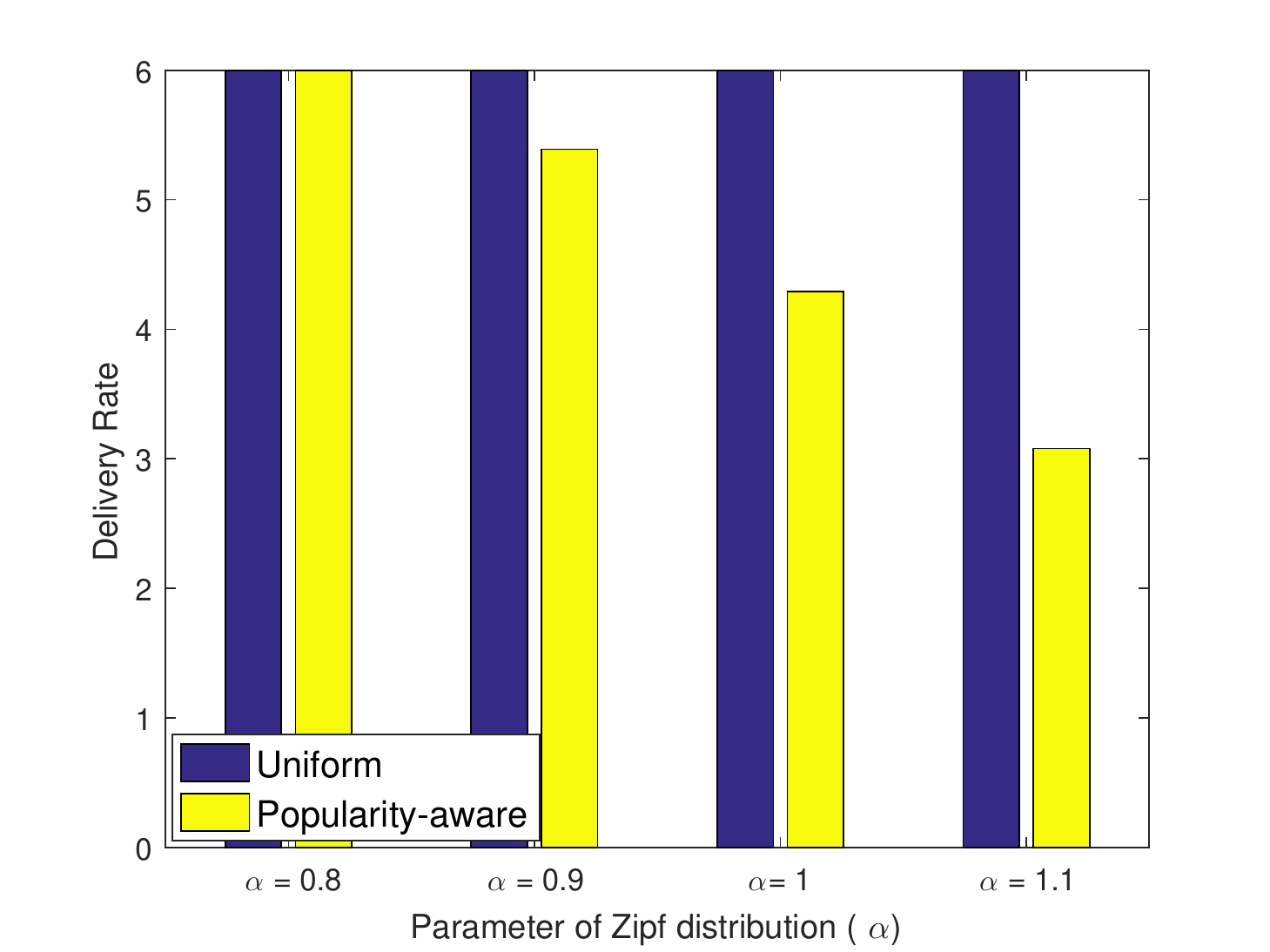}
					\caption{M=150, N=1200 and K=24} 
				\label{popularity1}
    \end{subfigure}
    \begin{subfigure}[b]{0.47\textwidth}
     \centering
        \includegraphics[scale=0.5]{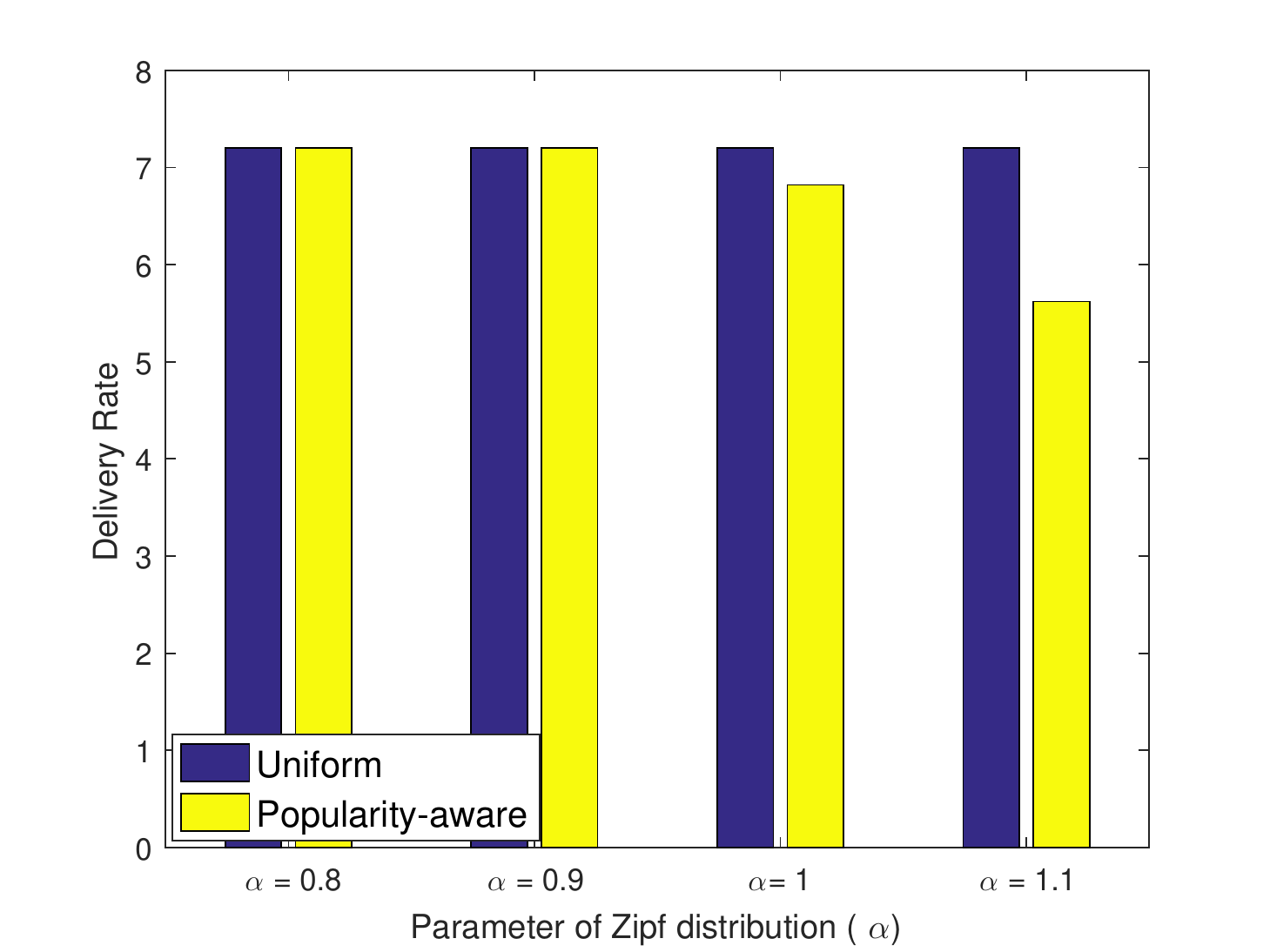}
        \caption{M=150, N=1200 and K=48} 
				\label{popularity2}
        \end{subfigure}
				\caption{Comparison of uniform caching with popularity-aware caching}
		\label{non-uniform}
\end{figure*}
\indent To analyze the performance of the popularity and mobility-aware caching scheme, we consider the following setup with cache memory size $M=150$, library size $N=1200$, session duration $T=2$, and number of clusters $L=3$. To model the request probability of the files in the library, we use Zipf distribution with parameter $\alpha$ which illustrates the skewness of the distribution. We consider two scenario with number of SBSs $K=24$ and $K=48$, and for each scenario we consider four different  $\alpha$ values $\left\{ 0.8,0.9,1,1.1 \right\}$. The expected delivery rate under uniform cache placement and popularity-aware cache placement schemes are illustrated in Fig. \ref{non-uniform}. In the first scenario, for K=24, we observe that up to $50\%$ reduction on the delivery rate is achievable by  employing mobility-aware scheme together with the popularity-aware placement  compared to the mobility-aware scheme with uniform cache placement. We also observe that the impact of  popularity-aware placement on the delivery rate is less visible in the second scenario with $K=48$. Therefore, we conclude that popularity-aware placement has more visible impact when $t$ is small.

\subsection{Multi-user scenario with extended mobility model}
To highlight  the main ideas of the proposed mobility-aware caching and delivery scheme, we have so far focused on a simple setting, in which there is exactly one MU connecting to each SBS. Nevertheless, the results can be extended to more general scenarios, with multiple MUs in each cell. One can easily observe that if there are exactly $Q$  MUs in each cell, then the delivery rate linearly scales with $Q$. Consequently, we would like to see how efficient  the proposed mobility-aware strategy is for general random mobility scenario, in which a MU can remain in the same cell for more than one time slot, and the number of MUs in each cell may not be equal. We remark that, SBSs, particularly those with small coverage areas, such as picocells, can serve limited number of users. To this end, we introduce SBS service capacity, $Q_{s}$, as the maximum number of MUs that can be served by a SBS in a time slot. Further, we introduce the term $\lambda$ to illustrate the user density, which denotes the ratio of the average number of MUs connecting to a SBS with service capacity $Q_{s}$.\\
\indent To measure the performance of the proposed mobility-aware scheme, we analyze the offloading rate of the network, which reflects how efficiently the SBSs are utilized. We define the offloading rate, $\sigma(K,T,Q_{s})$ as the ratio of the average number of fragments served by the SBSs during a session of $T$ time slots to the overall service capacity of the SBSs, given by $Q_{s}\times K\times T$. Hence, if each MU stays at most one time slot in a cell, and if there are at least $Q_{s}$ MUs in each cell at each time slot, then the SBSs are fully utilized and $\sigma(K,T,Q_{s})=1$.\\ 
\indent For the simulations, we consider a $6\times4$ square grid topology with $K=24$ SBSs, and set $Q_{s}=20$. In the simulation, we consider session duration of $T=4$ and $T=5$, and for each session duration we increase the user density parameter $\lambda$ from 1.25 to 2.25 with a step size of 0.25. To model the mobility of the MUs, different stochastic models such as Markov processes \cite{mobility1,mobility2,mobility3,mobility4} and Poisson processes \cite{G.M1,G.M6} can be employed. Markov process based models have been shown to provide accurate models, based on real mobility traces, particularly for vehicular users \cite{mobility2,mobility4}. Hence, we model MUs mobility pattern with a discrete time Markov process, whose state is the current cell location, and we assume uniform transition probabilities; that is, for the square grid topology, a MU stays in a cell with probability 0.2, and moves one of the 4 neighbouring cells with the same probability. Initially, we allocate $Q_{s}\times \lambda$ MUs in each cell and then let users move according to given discrete time Markov process. Simulation results for $T=4$ and $T=5$ are illustrated in Fig. \ref{offload}.\\
\indent  We note that there are two factors that can reduce the offloading rate of the mobility-aware scheme; first, the non-uniform distribution of the MUs; that is, if the number of MUs in a cell is less than $Q_{s}$ then this particular SBS will be underutilized; second, a MU cannot be served by the SBS more than one time slot due to the placement scheme used. The first factor is a common issue for the  performance of existing coded caching schemes; however, the second factor is particular to the proposed mobility-aware scheme. To measure the effect of the second factor we also measure the offloading rate of the mobility-aware scheme ignoring the limitation due to the placement phase, so that a SBS can serve more than one fragment. We observe that, as illustrated in Fig. \ref{offload}, mobility-aware scheme achieves an offloading rate close to $0.9$ in general. As expected, we also observe that the offloading rate increase with user density and decreases with $T$. Finally, another important observation is that relaxation of the limitation due to the placement phase (a SBS can serve exactly one fragment), has a negligible affect on the offloading rate. This observation simply implies that the first factor,the non-uniform distribution of MUs among the cells, is the dominant factor on the performance of the proposed mobility-aware scheme.
\begin{figure*}
         \begin{subfigure}[b]{0.47\textwidth}
          \centering
     \includegraphics[scale=0.5]{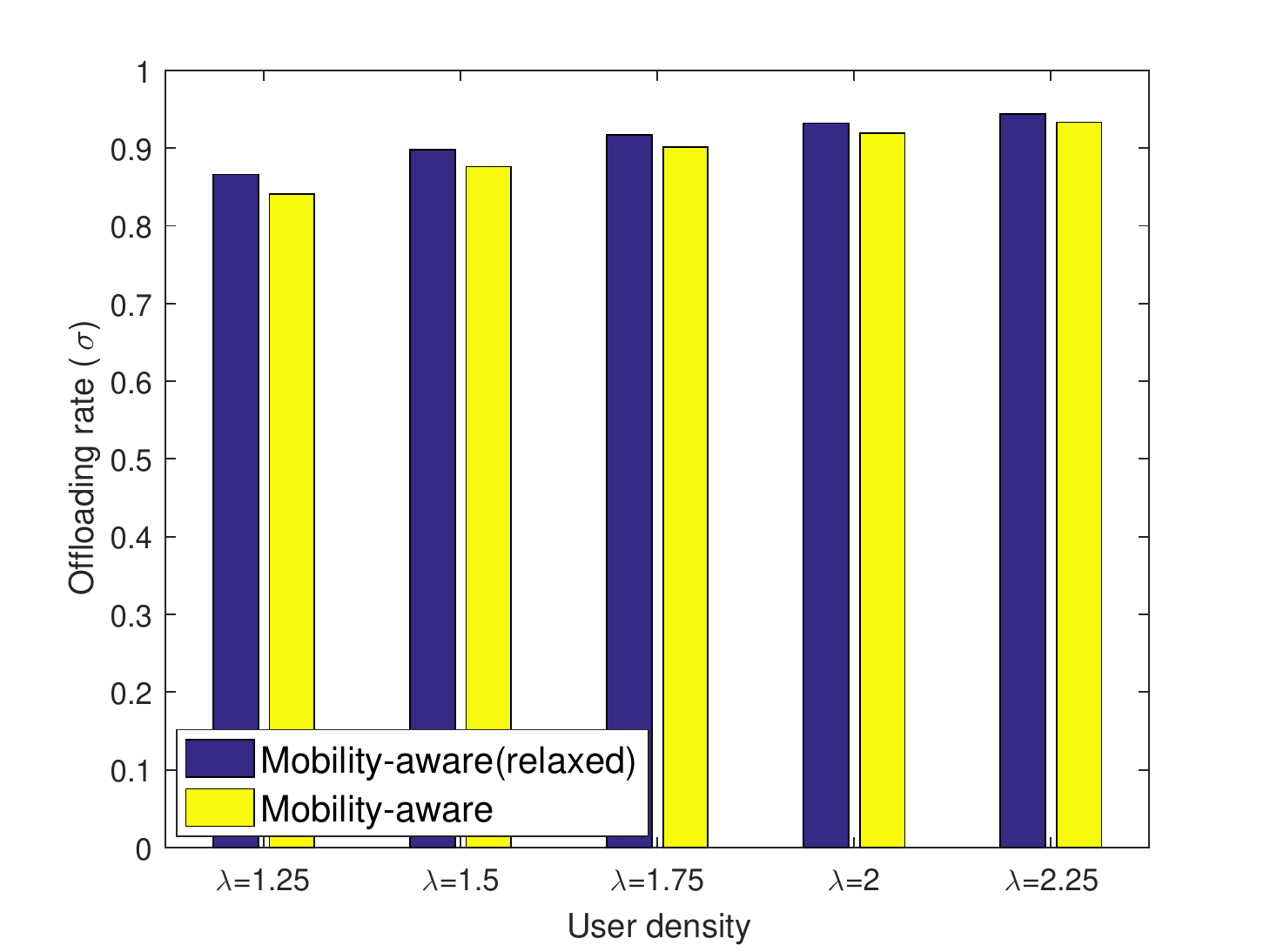}
					\caption{Service capacity $Q_{s}=20$ and download duration $T=4$.} 
				\label{offload1}
    \end{subfigure}
    \begin{subfigure}[b]{0.47\textwidth}
     \centering
        \includegraphics[scale=0.5]{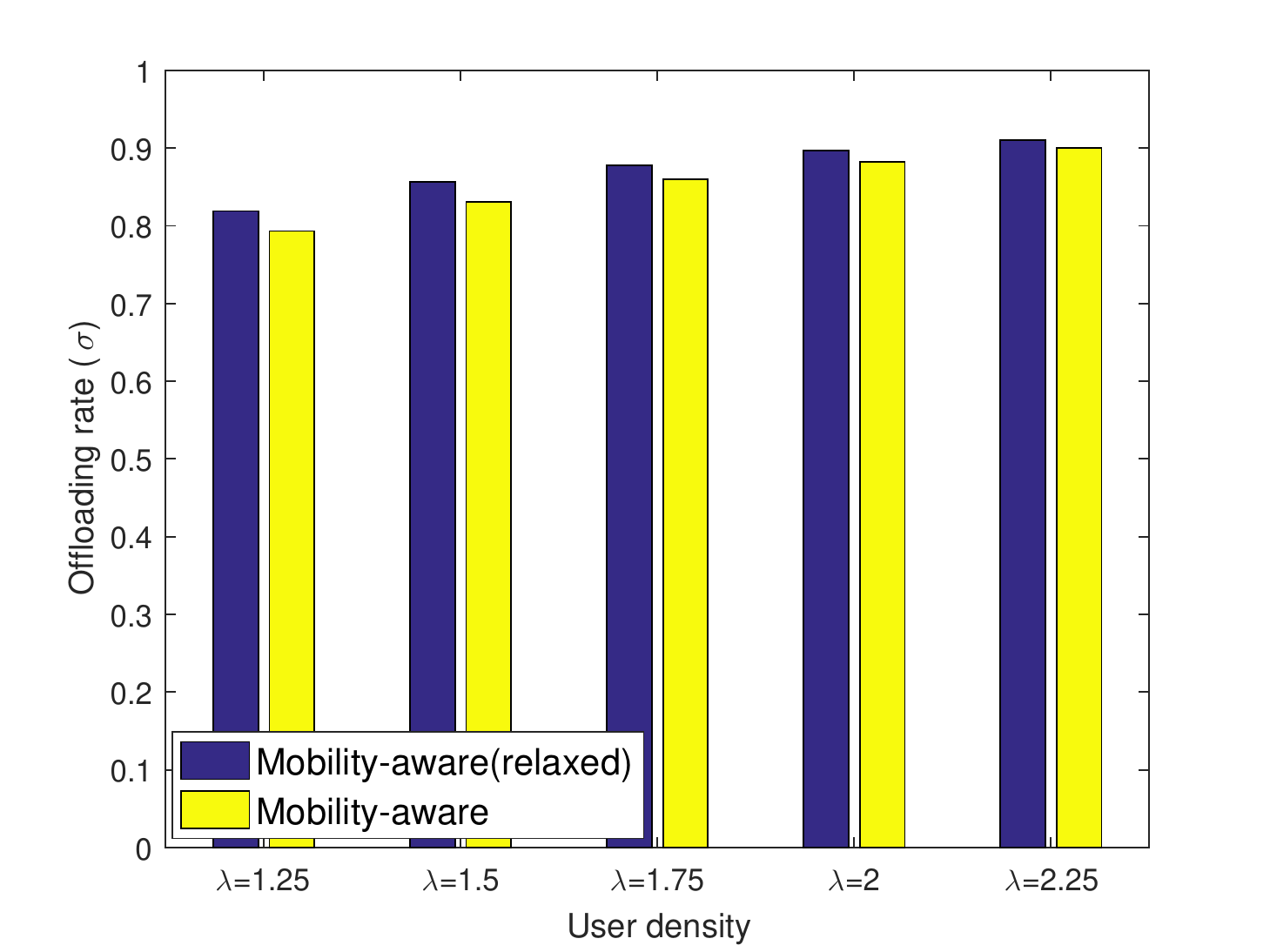}
        \caption{Service capacity $Q_{s}=20$ and download duration $T=5$.} 
				\label{offload2}
        \end{subfigure}
				\caption{Offloading rate, $\sigma(K,T,Q_{s})$, of the proposed mobility-aware scheme under random mobility scenario in a $6\times4$ square grid topology.}
		\label{offload}
\end{figure*}

 \section{Conclusions}\label{sec:conc}
We have introduced a novel MDS-coded storage and coded delivery scheme that adopts its caching strategy to the mobility patterns of the users. Our scheme exploits a coloring scheme for the SBSs, inspired by frequency reuse patterns in cellular networks, that have been extensively studied in the past to reduce interference. The files in the library are divided into sub-files, which are MDS-coded, and stored in the SBS caches, allowing users to satisfy their demands from multiple SBSs on their path under a high mobility assumption. We have shown that the proposed strategy achieves a significant reduction in the number of sub-files; particularly when the number of sub-files that can be created is limited, either due to the finite file size or to limit the complexity of the caching and delivery scheme, the proposed scheme provides significant gains in the backhaul load. We have also shown that the benefits of the proposed mobility-aware scheme extends also to non-uniform popularity distributions as well as to more general mobility scenarios allowing arbitrary random mobility patterns and multiple users being served by each SBS simultaneously.

%
\IEEEpeerreviewmaketitle

\bibliographystyle{IEEEtran}
\bibliography{IEEEabrv,ref}

\end{document}